\documentclass[pra, twocolumn, 11pt, tightenlines,
notitlepage, superscriptaddress, longbibliography]{revtex4-1}

\usepackage[utf8]{inputenc}
\usepackage[margin=2cm]{geometry}

\usepackage{amsmath, amsthm, amsfonts, amssymb, bm, bbm}
\usepackage{graphicx}
\usepackage{hyperref, cleveref}
\hypersetup{
	colorlinks,
	linkcolor={blue},
	citecolor={blue},
	urlcolor={blue}
}
\usepackage{enumerate}
\usepackage{ulem}
\usepackage{placeins}
\usepackage{mathtools}

\usepackage{pgfplots}
\theoremstyle{plain}
\newtheorem{thm}{Theorem}

\newtheorem{proposition}[thm]{Proposition}

\newtheorem{cor}[thm]{Corollary}

\theoremstyle{definition}

\DeclareMathOperator{\tr}{Tr}

\newcommand{\bb}{\mathbb}
\newcommand{\mc}{\mathcal}
\newcommand{\mf}{\mathfrak}

\newcommand{\ct}{\ensuremath{^\dagger}}

\newcommand{\ket}[1]{\ensuremath{|#1\rangle}}
\newcommand{\bra}[1]{\ensuremath{\langle#1|}}
\newcommand{\dket}[1]{\ensuremath{|#1\rangle\!\rangle}}
\newcommand{\dbra}[1]{\ensuremath{\langle\!\langle#1|}}

\newcommand{\drangle}{\ensuremath{\rangle\!\rangle}}

\makeatletter
\def\ketbra#1{\def\tempa{#1}\futurelet\next\ketbra@i}% Save first argument
\def\ketbra@i{\ifx\next\bgroup\expandafter\ketbra@ii\else\expandafter\ketbra@end\fi}%Check brace
\def\ketbra@ii#1{\ket{\tempa}\bra{#1}}%Two args
\def\ketbra@end{\ket{\tempa}\bra{\tempa}}%Single args
\makeatother

\makeatletter
\def\dketbra#1{\def\tempa{#1}\futurelet\next\dketbra@i}% Save first argument
\def\dketbra@i{\ifx\next\bgroup\expandafter\dketbra@ii\else\expandafter\dketbra@end\fi}%Check brace
\def\dketbra@ii#1{\dket{\tempa}\dbra{#1}}%Two args
\def\dketbra@end{\dket{\tempa}\dbra{\tempa}}%Single args
\makeatother

\makeatletter
\def\dbraket#1{\def\tempa{#1}\futurelet\next\dbraket@i}% Save first argument
\def\dbraket@i{\ifx\next\bgroup\expandafter\dbraket@ii\else\expandafter\dbraket@end\fi}%Check brace
\def\dbraket@ii#1{\dbra{\tempa}#1 \drangle}%Two args
\def\dbraket@end{\dbra{\tempa}\tempa\drangle}%Single args
\makeatother

\usepackage{fdsymbol}

\hypersetup{pdftitle={Code symmetries},
	pdfauthor={Stefanie Beale}}% add a title to the metadata

\begin{document}
	
	\title{Efficiently computing logical noise in quantum error-correcting codes}
	\author{Stefanie J. Beale}
	\email{sbeale@uwaterloo.ca}
	\affiliation{ Institute for Quantum Computing and Department of Physics and
		Astronomy, University of Waterloo, Waterloo, Ontario, N2L 3G1, Canada}
	\affiliation{Quantum Benchmark Inc., 51 Breithaupt Street, Suite 100 Kitchener, Ontario N2H 5G5, Canada}
	
	\author{Joel J. Wallman}
	\email{jwallman@uwaterloo.ca}
	\affiliation{Quantum Benchmark Inc., 51 Breithaupt Street, Suite 100 Kitchener, Ontario N2H 5G5, Canada}
	\affiliation{ Institute for Quantum Computing and Department of Applied Mathematics, University of Waterloo, Waterloo, Ontario, N2L 3G1, Canada}
	
	\begin{abstract}
		Quantum error correction protocols have been developed to offset the high sensitivity to noise inherent in quantum systems.
		However, much is still unknown about the behaviour of a quantum error-correcting code under general noise, including noisy measurements.
		This lack of knowledge is largely due to the computational cost of simulating quantum systems large enough to perform nontrivial encodings.
		In this paper, we develop general methods for incorporating noisy measurement operations into simulations of quantum error-correcting codes and show that measurement errors on readout qubits manifest as a renormalization on the effective logical noise.
		We also derive general methods for reducing the computational complexity of calculating the exact effective logical noise by many orders of magnitude.
		This reduction is achieved by determining when different recovery operations produce equivalent logical noise.
		These methods could also be used to better approximate soft decoding schemes for concatenated codes or to reduce the size of a lookup table to speed up the error correction step in implementations of quantum error-correcting codes.
		We give examples of such reductions for the three-qubit, five-qubit, Steane, concatenated, and toric codes.
	\end{abstract}

	
	\maketitle
	
	\section{Introduction}
	
	Quantum error-correcting codes~\cite{Gottesman2010, Knill1997} (QECCs) and other methods for fault-tolerant quantum computation will likely be required to use quantum computers to solve otherwise intractable problems.
	However, determining the performance of QECCs is difficult for precisely the same reason that we want to develop them, namely, that simulating large-scale quantum systems is computationally expensive.
	The difficulty of simulating the performance of QECCs is compounded by the fact that the effective noise process depends on the specific syndromes that are observed, where the number of possible syndromes typically grows exponentially with the number of physical qubits and rounds of computation.
	Consequently, little is known about the behaviour of QECCs under generic noise processes.
	
	In recent years, there has been a resurgence of interest in studying the behaviour of general and specific noise in QECCs \cite{Gutierrez2016, Greenbaum2017, Darmawan2017, Chamberland2017, Bravyi2017, Beale2018, Huang2019, Cai2019, Iverson2019}.
	There has also been interest in studying symmetries in QECCs for several purposes, including noise tailoring \cite{Cai2019}, code construction via classical cyclic codes \cite{LaGuardia2017}, applying logical operations \cite{Grassl2013, Chao2018}, and reducing the complexity of simulations \cite{Huang2019}.
	In this paper, we extend the previous formalism used to study quantum memories~\cite{Rahn2002} to include noisy measurements.
	We then derive general conditions under which two syndromes will result in equivalent logical noise in a QECC.
	These degeneracies can reduce the computational cost of studying QECCs by orders of magnitude, and once degeneracies are found for a given QECC, they can be re-used for any future simulation of that code under any applicable noise model without any additional computation.
	Moreover, part of the motivation for simulating QECCs is to determine good choices of recovery maps for each syndrome.
	By establishing general and simple conditions under which recovery maps associated with different syndromes are degenerate, we can reduce the number of syndromes for which good choices need to be cached or have complex calculations performed.
	We anticipate that this will enable faster implementations of decoders, especially in memory-constrained environments (such as cryogenic control computers or decoders built into field-programmable gate arrays), so errors can be corrected before they cascade.
	
	In \Cref{sec:qec} of this paper, we review QECCs in a general setting and extend the formalism of Reference~\cite{Rahn2002} to include noisy quantum measurements.
	In \Cref{sec:equivalentConditionalMaps}, we derive general results showing when the effective noise conditioned upon two different measurement outcomes is equivalent.
	In \Cref{sec:stabSym}, we demonstrate that noisy readout measurements in stabilizer codes leave the effective logical noise invariant up to a renormalization and show how our results can be applied to reduce the simulation cost by applying our results to several small stabilizer codes, the general toric code, and concatenated codes.
	Even at the first level of concatenation, our results can reduce the exact simulation cost of a soft decoder (see \cite{Poulin2006} for a description of soft decoding) by a factor of $64^7/34 992 \approx 10^8$ for depolarizing noise in the Steane code.
	Similar results were obtained for some of these codes under a more restricted noise model (namely, noise models with a single Kraus operator) in Reference~\cite{Huang2019}, however, the results here are broader (they identify more degenerate syndrome maps), more general (they apply to general codes and noise maps, including noisy readout measurements), and allow symmetry operations to be verified with ease.
	
	\section{Quantum error-correcting codes}\label{sec:qec}
	
	In \Cref{sec:prelim}, we review quantum measurements, emphasizing how they can be regarded as sets of linear maps up to a benign renormalization factor.
	In \Cref{sec:qecsteps}, we describe a quantum error correction step in terms of the normalized measurement maps set out in \Cref{sec:prelim}.
	
	We use the following general notation throughout this paper.
	For any Hilbert spaces $\bb{H}$ and $\bb{K}$, let $\bb{B}(\bb{H}, \bb{K})$ denote the space of bounded linear maps from $\bb{H}$ to $\bb{K}$, $\textrm{Pos}(\bb{H})\subset\bb{B}(\bb{H},\bb{H})$ be the set of bounded and positive semi-definite operators from $\bb{H}$ to itself, and let $\bb{U}(\bb{H}, \bb{K}) \subset \bb{B}(\bb{H}, \bb{K}) / \bb{U}(1)$ denote the set of isometries from $\bb{H}$ to $\bb{K}$, where we remove global phases from the set of isometries because they are unobservable in quantum mechanics.
	When $\bb{H}=\bb{K}$, we use the shorthands $\bb{B}(\bb{H}) = \bb{B}(\bb{H}, \bb{H})$ and $\bb{U}(\bb{H}) = \bb{U}(\bb{H}, \bb{H})$.
	For clarity, we will distinguish between operators, that is, elements of $\bb{B}(\bb{H}, \bb{K})$, and superoperators, that is, elements of $\bb{B}(\bb{B}(\bb{H}), \bb{B}(\bb{K}))$, although both are formally linear maps.
	Specifically, we will exclusively use Roman font or Greek letters for operators and calligraphic font for superoperators.
	As a special case, any operator $A\in\bb{B}(\bb{H},\bb{K})$ defines a superoperator $\mc{A}\in\bb{B}(\bb{B}(\bb{H}), \bb{B}(\bb{K}))$ acting by conjugation, that is, $\mc{A}(M)=AMA\ct\: \forall M\in\bb{H}$.
	Note that the map $A \to \mc{A}$ is a bijection up to a global (and irrelevant) phase, so we abuse notation by interchanging $A$ and $\mc{A}$.
	Also note that the maps $A \leftrightarrow \mc{A}$ preserve multiplication, that is, $AB \leftrightarrow \mc{A} \mc{B}$.
	
	\subsection{Quantum measurements}\label{sec:prelim}
	
	A projection-valued measure (PVM) $\mf{A} \subset \mathrm{Pos}(\bb{H})$ is a set of orthogonal projectors that sum to the identity, that is, $A^2 = A$ for all $A \in \mf{A}$ and $\sum_{A \in \mf{A}} A = 1_\bb{H}$ where $1_\bb{H}$ denotes the identity element of $\bb{B}(\bb{H})$.
	Ideally, the state after applying a PVM $\mf{A}$ to a system in the state $\rho$ and observing outcome $A$ is
	\begin{align}\label{eq:postMeasurement}
	\tau_{\rho, A} = \frac{\mc{A}(\rho)}{\tr \mc{A}(\rho)},
	\end{align}
	where $\tr \mc{A}(\rho)$ is the probability of observing the outcome $A \in \mf{A}$ conditioned on the system being in the state $\rho$.
	However, experimental imperfections cause realistic measurements to deviate from an ideal projective measurement in a number of ways.
	First, noise may result in superoperators being applied to the system before and/or after the ideal measurement.
	Second, the measurement procedure may not be describable by a PVM, but rather by a positive-operator-valued measure (POVM), which is a subset $\mf{A} \subset \mathrm{Pos}(\bb{H})$ such that $\sum_{B \in \mf{A}} B = 1_\bb{H}$ (that is, we relax the assumption that the elements of a PVM are projectors).
	\Cref{eq:postMeasurement} can be directly generalized to POVMs using the decomposition $B=A\ct A\: \forall B\in\mf{A}\subset\textrm{Pos}(\bb{H})$ to define an operator $A$ (which need not be positive and is not unique).
	This reduces to the standard case because $B^\dagger B = B$ for any projector $B$.
	
	Crucially, for both PVMs and POVMs, the post-measurement state is not a linear function of the input state.
	However, the nonlinearity is of a benign form, namely, the nonlinearity from conditioning upon the output state that arises from Bayes' rule.
	With some abuse of notation (namely, interpreting 0/0 as 1), the expected post-measurement state is
	\begin{align}
	\sum_{A \in \mf{A}} \tau_{A, \rho} \tr \mc{A}(\rho) = \sum_{A \in \mf{A}} \mc{A}(\rho).
	\end{align}
	To accommodate additional noise processes, we can simply replace the implict linear maps $\mc{A}(\rho) = A \rho A^\dagger$ by general linear maps that can include pre- and post measurement control operations and noise maps.
	That is, we define a general measurement $\mf{M}$ to be a set of superoperators, that is, $\mf{M} \subset \bb{B(B(H), B(K))}$, although we will only consider the case $\bb{K} = \bb{H}$.
	To be physical, each $\mc{M} \in \mf{M}$ must map valid quantum states to valid quantum states (that is, be completely positive but generally not trace preserving).
	To be a complete measurement, some outcome must always occur if the input state is a density matrix and so
	\begin{align}\label{eq:fullMap}
	\sum_{\mc{M} \in \mf{M}} \mc{M}
	\end{align}
	must be a trace-preserving map.
	Note that any completely positive map can be represented in Kraus operator form as a channel in the form of \Cref{eq:fullMap}, and so noise processes can formally be regarded as a measurement where the outcome is not recorded.
	
	\subsection{Quantum error correction}\label{sec:qecsteps}
	
	An encoding of $\bb{H}$ into $\bb{K}$ is an isometry $U \in \bb{U}(\bb{H},\bb{K})$.
	For any encoding, we can choose a set of recovery maps $\mf{R} \subset \bb{U}(\bb{K})$ such that the set of projectors $\{\Pi_R=\mc{RUU\ct R\ct} : R \in \mf{R}\}$ associated with the cospaces of the code, form a projective measurement.
	That is, the recovery maps satisfy $U\ct Q\ct R U = \delta_{Q, R} 1_\bb{H}$ for all $Q, R \in \mf{R}$ and $\sum_{R \in \mf{R}} RU U\ct R\ct = 1_\bb{K}$.
	For example, the three-qubit repetition code encodes one logical qubit into three physical qubits via the isometry $U = \ketbra{000}{0} + \ketbra{111}{1}$.
	We can extend this isometry into a projective measurement on the encoded space by choosing, for example, $\mf{R} = \{III, XII, IXI, IIX\}$.
	The elements of $\mf{R}$ can be regarded as the most likely errors to occur.
	Thus, if a system is prepared in the encoded state $U\psi$ and a specific error $E \in \mf{R}$ occurs, the ideal PVM $\{\Pi_R=\mc{R U U\ct R}\ct : R \in \mf{R}\}$ will result in the outcome $E$, which can then be corrected by applying $\mc{E}\ct$.
	Therefore, we fold the recovery maps $\mc{R}\ct$ for each cospace into the measurement for convenience and define the ideal measurement to be $\mf{M} = \{\mc{M}\}$, where 
	$\mc{M}=\mc{R}\ct\Pi_R$ so that $\mf{M}=\{\mc{U U\ct R\ct} : R \in \mf{R}\}$. We follow this convention throughout this paper.
	A QECC is then a pair $(U, \mf{M})$, and can be used to protect a logical qubit against a noisy physical process $\mc{N}$ (typically a completely positive, trace-preserving map) as follows:
	\begin{enumerate}[ {(}1{)} ]
		\item Choose an input state $\bar{\rho}\in\rm{Pos}(\bb{H})$ with $\tr\bar{\rho}=1$.
		\item Prepare the state $\rho = \mc{U}(\bar{\rho})$.
		\item Send the state through a noisy channel $\mc{N} \in \bb{B}(\bb{B}(\bb{K}))$.
		\item Perform a measurement $\mf{M}$.
		\item Apply the decoding map $\mc{U}\ct$.
	\end{enumerate}
	While the above procedure includes applying the decoding map, in practice one would typically measure the expectation values of encoded operators or treat the output as an encoded input into a subsequent round of error correction, where subsequent operations can be conditioned upon the observed outcome of $\mf{M}$.
	
	When the implementation of the encoding $U$ is noisy, the encoded state can end up outside of the codespace, in which case, we lose information about the encoded state.
	Often, a gate will be applied before the correction step (step 3); in this case, $\mc{N}$ is replaced with a possibly noisy implementation of a gate or a fault-tolerant gadget.
	
	When the outcome $\mc{M} \in \mf{M}$ is observed, the above process results in the conditional map
	\begin{align}\label{eq:reducedMap}
	\bar{\mc{N}}_U(\mc{M}) = \mc{U}\ct\mc{M} \mc{N} \mc{U},
	\end{align}
	where we do not divide by the probability with which the outcome occurs.
	We define the average logical channel as the average over the conditional maps:
	\begin{align}
	\bar{\mc{N}}_U = \sum_{\mc{M} \in \mf{M}} \bar{\mc{N}}_U(\mc{M}).
	\end{align}
	The average logical channel is often used to benchmark the performance of a given QECC $(U, \mf{M})$ against a given noise model $\mc{N}$.
	
	For ideal measurements, the ideal recovery operator uniquely specifies the measurement outcome.
	Therefore, we define the effective map conditioned on the recovery map R to be the effective map conditioned on $\mc{M}_R = \mc{UU^\dagger R^\dagger}$, following our convention of folding the recovery map into the measurement.
	That is, we define
	\begin{align}\label{eq:idealReducedMap}
	\mc{\dot{N}}_U(R) =\bar{\mc{N}}_U(\mc{M}_R)= \mc{U}\ct \mc{R}\ct \mc{N} \mc{U},
	\end{align}
	as in Reference~\cite{Rahn2002}, where we have used the fact that $U\ct U U\ct = U\ct$ for the isometry $U$ and we use $\dot{\mc{N}}$ and Roman font rather than $\bar{\mc{N}}$ and caligraphic font to differentiate between conditioning on a measurement or a recovery operation.
	
	To model noisy readout measurements in QECCs, we assume that the effective measurement on the encoded space can be represented as a probabilistic sum over ideal projectors onto the cospaces of the code. That is, for a measurement outcome associated with recovery operation $R$, the noisy measurement channel $\tilde{\Pi}_R$ is given by
	
	\begin{align}
		\tilde{\Pi}_R(\rho)=\sum_{Q\in\bb{R}}\bar{C}_{Q,R}\Pi_Q(\rho),
	\end{align}
	where $\bar{C}_{Q,R}$ is an element of the logical $2^{n-k}\times 2^{n-k}$ confusion matrix, $\bar{C}$.
	One method used to protect against measurement errors of this form is to measure each stabilizer generator multiple times and correct according to the outcome that occurs most frequently.
	We prove that this assumption is valid for stabilizer codes with ideal syndrome extraction circuits and noisy measurement of the readout qubits in \Cref{sec:stabSym}.
	To be consistent with the effective logical map defined for the case where measurement is ideal, we fold the recovery operation into the measurement:
	
	\begin{align}
	\tilde{\mc{M}}=\mc{R}\ct\tilde{\Pi}_R.
	\end{align}
	Note that because ideal measurements are a special case of noisy measurements where $\bar{C}_{j,m}=\delta_{j,m}$, any results derived using noisy measurements of the form $\tilde{\mc{M}}$ throughout this paper hold equally for ideal measurements.
	We denote the effective logical channel with noisy measurements the same as that for ideal measurements but with the addition of a $\: \tilde{}\: $ over $\dot{\mc{N}}$ or $\bar{\mc{N}}$. We can then write
	
	\begin{align}
	\tilde{\dot{\mc{N}}}_U(R)&=\tilde{\bar{\mc{N}}}_U(\tilde{\mc{M}})\\
	&=\mc{U}^\dagger\tilde{\mc{M}}\mc{N}\mc{U}\\
	&=\mc{U}^\dagger\mc{R}^\dagger\tilde{\Pi}_R\mc{N}\mc{U}\\
	&=\sum_{Q\in\bb{R}}\bar{C}_{Q,R} \mc{U}^\dagger\mc{R}^\dagger\Pi_Q\mc{N}\mc{U}\label{eq:confmatsum}\\
	&=\bar{C}_{R,R}\dot{\mc{N}}_U(R).
	\label{eq:noisymeasptm}
	\end{align}
	The simplification from \Cref{eq:confmatsum} to \Cref{eq:noisymeasptm} is possible because we have the decoding map, $\mc{U}\ct$, after the measurement so any terms not in the codespace after the measurement are removed by $\mc{U}\ct$.
	In the case where the state is not decoded after the error correction step, the cospaces will remain populated from cases where $\delta_{m,j}=0$.
	 Examining \Cref{eq:noisymeasptm}, we see that the effective logical channel with noisy measurements is identical to the effective logical channel with ideal measurement, up to a renormalization.
	For the remainder of this paper, when we refer to noisy measurements we mean measurements with noise of the form described in this section.
	
	%In general, the output state is not a tensor product state, which will affect subsequent measurements.
	%Two ways we could enforce a tensor product output are to discard the readout system or to reset it to a fixed state.
	%If we discard the readout system then the effective measurement on the encoded system is $\{\tr_\bb{T} \mc{A} \mc{E} : \mc{A} \in \mf{A}\}$ where $\tr_\bb{T}\in\bb{B}(\bb{K}\otimes\bb{T},\bb{K})$ denotes the linear map whose output is the partial trace of the input state over $\bb{T}$.
	%Resetting the readout system to a fixed state $\gamma \in \rm{Pos}(\bb{T})$ corresponds to applying some map $I \otimes \mc{B}\in\bb{B}(\bb{K}\otimes\bb{T})$ to create an effective measurement $\{(I\otimes \mc{B})\mc{A}\mc{E}:\mc{A}\in\mf{A}\} \subset\bb{B}(\bb{K}\otimes\bb{T})$, where $\mc{B}(\sigma) = \gamma \tr(\sigma)$ for all $\sigma \in \bb{B}(\bb{T})$.
	%The channel $I \otimes \mc{B}$ is then an entanglement breaking channel and so we obtain a product state output after the measurement.
	%As per our convention, we fold the (potentially noisy) recovery operation $R^\dagger$ associated with a given outcome into the measurements and define the following measurements:
	%\begin{align}
	%&\mf{M}_{discard}=\{\mc{R}\ct\tr_\bb{T} \mc{A} \mc{E} : \mc{A} \in \mf{A}\}\textrm{ and }\\
	%&\mf{M}_{reset}=\{\mc{R}\ct(I\otimes \mc{B})\mc{A}\mc{E}:\mc{A}\in \mf{A}\}.
	%\end{align}
	%To model an error correcting code with a noisy measurement, we can thus choose $\mf{M}$ to be either $\mf{M}_{discard}$ or $\mf{M}_{reset}$ in step 4 of the error correction protocol discussed above.
	
	\section{Simulating quantum error correcting codes}\label{sec:equivalentConditionalMaps}
	
	The above error correction procedure is defined to be a pair $(U, \mf{M})$, where $\mf{M}$ is the measurement including recovery operators and we often label measurement outcomes by the associated recovery map.
	To simulate the error correction procedure, we need to compute the conditional maps $\bar{\mc{N}}_U(\mc{M})$.
	There are many maps and each conditional map is typically expensive to compute.
	However, as observed in \cite{Chamberland2017, Huang2019}, many of the conditional maps are related, which can be exploited to reduce the computation time.
	We define two superoperators $\mc{L,M} \in \bb{B(B(H))}$ to be degenerate if $\mc{L} = c\mc{M}$ for some $c\in\bb{R}$ and to be logically degenerate if there exist invertible superoperators $\bar{\mc{A}}, \bar{\mc{B}} \in \bb{B(B(H))}$ such that
	\begin{align}
	\mc{L} = c\mc{\bar{A} M \bar{B}}
	\end{align}
	for some $c\in\bb{R}$.
	We are primarily interested in the superoperators associated to measurement outcomes, and so we say that two outcomes $\mc{Q}, \mc{R} \in \mf{M}$ are (logically) degenerate for a fixed noise process $\mc{N}$ if the corresponding superoperators $\bar{\mc{N}}_U(\mc{Q})$ and $\bar{\mc{N}}_U(\mc{R})$ are (logically) degenerate.
	We say that two measurement outcomes are nondegenerate for a fixed noise process $\mc{N}$ if they are not known to be degenerate under $\mc{N}$.
	We make this particular distinction because we may not know all degeneracies for a class of noise processes.
	We call a set of degenerate measurement outcomes a degeneracy class and a set of logically degenerate measurement outcomes a logical degeneracy class.
	
	Some degeneracy relations are easily established for any noise process using stabilizers and logical operators.
	Stabilizers and logical operators of an encoding $U$ are invertible superoperators $\mc{S}, \mc{L} \in \bb{B}(\bb{B}(\bb{K}))$ such that $\mc{S} \mc{U} = \mc{U}$ and $\mc{L} \mc{U} = \mc{U}\bar{\mc{L}}$ for some $\bar{\mc{L}} \in \bb{B}(\bb{B}(\bb{H}))$, respectively.
	A stabilizer is a special case of a logical operator where $\bar{\mc{L}} = 1_{\bb{B}(\bb{H})}$.
	For example, the superoperators corresponding to $ZZI$ and $ZZZ$ are a stabilizer and a nontrivial logical operator of $U = \ketbra{000}{0} + \ketbra{111}{1}$, respectively.
	The stabilizer and logical groups are the groups $\mf{S}(U)$ of stabilizer and $\mf{L}(U)$ of logical operations, respectively.
	For any stabilizer $\mc{S} \in \mf{S}(U)$ and logical operator $\mc{L} \in \mf{L}(U)$, we have $\bar{\mc{N}}_U(\mc{RS}) = \bar{\mc{N}}_U(\mc{R})$ and $\bar{\mc{N}}_U(\mc{LR}) = \bar{\mc{L}} \bar{\mc{N}}_U(\mc{R})$.
	Therefore recovery maps in the same left coset of $\mf{S}(U)$ will be degenerate and recovery maps in the same left coset of $\mf{L}(U)$ will be logically degenerate.
	The above observation can be used to change a single measurement outcome $\mc{M}$ to make $\bar{\mc{N}}_U(\mc{M})$ closer to a given logical operation (in particular, the identity operation).
	However, two distinct elements $Q, R \in \mf{R}$ cannot be related by a logical operation, as otherwise they would violate the assumption that $U\ct Q\ct R U = 0$, that is, that the associated measurement operators are orthogonal.
	Therefore the relationship $\bar{\mc{N}}_U(\mc{LR}) = \bar{\mc{L}} \bar{\mc{N}}_U(\mc{R})$ cannot speed up the computation of the full set of conditional maps for a fixed measurement.
	Before proceeding any further, we prove that no distinct outcomes are logically degenerate for all noise processes under ideal or noisy measurements, and hence any degeneracies can only hold for restricted noise models.
	
	\begin{thm}\label{thm:noGlobal}
		For any encoding $U \in \bb{U(H)}$ and any two recovery maps $Q, R \in \mf{R}$ associated to different (possibly noisy) measurement outcomes, there exist noise models $\mc{N}$ such that $\mc{Q}$ and $\mc{R}$ are not degenerate under $\mc{N}$.
	\end{thm}
	
	\begin{proof}
		Let $\mc{N} = \mc{Q}$, so
		\begin{align}
		\tilde{\dot{\mc{Q}}}_U(Q) = \bar{C}_{Q,Q}\dot{\mc{Q}}_U(Q)=\bar{C}_{Q,Q}\mc{U\ct Q\ct Q U} = \bar{C}_{Q,Q} \bar{\mc{I}}
		\end{align}
		Recalling that $U\ct R\ct Q U=\delta_{Q, R} I_\bb{H}$, we have
		\begin{align}
		\tilde{\dot{\mc{Q}}}_U(R) = \bar{C}_{R,R}\dot{\mc{Q}}_U(R) =\bar{C}_{R,R} \mc{U\ct R\ct Q U}=\bar{\bf{0}}.
		\end{align}
		Therefore, for any invertible superoperators $\mc{\bar{A}, \bar{B}} \in \bb{B(B(H))}$,
		\begin{align}
		\mc{\bar{A}\tilde{\dot{Q}}}_U(R) \mc{\bar{B}} = \bar{\bf{0}} \neq \tilde{\dot{\mc{Q}}}_U(Q)
		\end{align}
		as required.
	\end{proof}
	
	Despite the apparently strong statement of \Cref{thm:noGlobal}, it has been observed that different measurement outcomes can be degenerate for broad families of noise processes~\cite{Chamberland2017,Huang2019}.
	Reference~\cite{Huang2019} gave conditions based on code symmetries to identify degenerate syndrome maps for independent and identically distributed (IID) unitary noise in stabilizer codes, where a noise channel $\mc{N}$ is IID if it can be written as $\mc{N} = \mc{N}_1^{\otimes n}$ for some noise process, $\mc{N}_1$, acting on a single system.
	We present a more general result with a trivial and constructive proof, giving conditions under which different recovery operations or measurements are \textit{logically} degenerate for general QECCs.
	These conditions hold for more general (including non-IID) noise.
	Our results also show how logical operations can be factored into the error correction step by updating ideal recovery operations to other operations which are logically degenerate to their ideal counterparts.
	The following proposition follows directly from the definitions laid out above.
	This result seems trivial in light of our notation, however, arriving at such conclusions in less abstract settings (e.g., for stabilizer codes) is quite challenging.
	
	\begin{proposition}\label{thm:theMainThing}
		Let $U\in\bb{U(H, K)}$, $\mc{M}\in\mf{M}$ be a measurement outcome, and $\mc{N}\in\bb{B(B(K))}$. For any $\mc{A, B} \in\mf{L}(U)$, the maps $\mc{U\ct B}^{- 1}\mc{ \tilde{M} N A U}$ and $\mc{U\ct \tilde{M} N U}$ are logically degenerate.
	\end{proposition}
	
	Note that the the symmetry operators $\mc{A, B}\in\mf{L}(U)$ are not required to be unitary, but in the case that $\mc{B}$ is unitary, we have $\mc{B}^{-1}=\mc{B}^\dag$.
	From \Cref{thm:theMainThing}, we can immediately identify some degenerate measurement outcomes when there are logical operators that commute with the noise.
	Specifically, we define the logical symmetry group of an encoding $U$ under a noise process $\mc{N}$ to be the group $\mf{L}(U, \mc{N})=\{\mc{G}\in\mf{L}(U): [\mc{N},\mc{G}]=0\}$, and the stabilizer symmetry group to be $\mf{S}(U, \mc{N}) = \mf{L}(U, \mc{N}) \cap \mf{S}(U)$.
	
	\begin{cor}\label{thm:symWithMeas}
		Let $U \in \bb{U(H, K)}$, $\tilde{\mc{M}} \in \tilde{\mf{M}}$, and $\mc{N} \in \bb{B(B(K))}$.
		The maps $\{\mc{\tilde{\bar{N}}}_U(\mc{B}^{-1} \mc{\tilde{M} A}) : \mc{A} \in \mf{L}(U, \mc{N}), \mc{B} \in \mf{L}(U)\}$ are logically degenerate.
		Furthermore, the maps $\{\mc{\tilde{\bar{N}}}_U(\mc{B}^{-1}\mc{\tilde{M} A}) : \mc{A} \in \mf{S}(U, \mc{N}), \mc{B} \in \mf{S}(U)\}$ are degenerate.
	\end{cor}
	
	\begin{proof}
		Let $\mc{A} \in \mf{L}(U, \mc{N})$ and $\mc{B} \in \mf{L}(U)$.
		By assumption, there exist invertible $\mc{\bar{A}, \bar{B}} \in \bb{B(B(H))}$ such that $\mc{AU} = \mc{U \bar{A}}$ and $\mc{BU} = \mc{U \bar{B}}$.
		Therefore, by \Cref{eq:reducedMap,eq:noisymeasptm} we have
		\begin{align*}
		\tilde{\bar{\mc{N}}}(\mc{B}^{-1} \mc{\tilde{M} A})&=\bar{C}_m\bar{\mc{N}}_U(\mc{B}^{-1}\mc{ M A}) \\
		&= \bar{C}_m\mc{U\ct B}^{-1}\mc{M A N U}\\
		& = \bar{C}_m\mc{U\ct B}^{-1}\mc{ M N A U} \\
		&= \bar{C}_m\mc{\bar{B}}^{-1}\mc{ U\ct M N U \bar{A}}\\
		& = \bar{C}_m\bar{\mc{B}}^{-1} \bar{\mc{N}}_U(\mc{M}) \bar{\mc{A}},
		\end{align*}
		where $\bar{C}_m$ is the diagonal element of the logical confusion matrix associated with measurement outcome $\mc{M}$, as required.
		Note that the logical group and the logical symmetry group are groups, and so all pairs in the first set are logically degenerate.
		The final statement holds because $\mc{\bar{A} = \bar{I} = \bar{B}}$ if $\mc{A, B} \in \mf{S}(U)$ by definition.
	\end{proof}
	
	Measurement is often assumed to be noiseless in studies of quantum computing, and, in fact, this assumption was made in each of the papers which observed symmetries in QECCs before this paper \cite{Chamberland2017,Huang2019}.
	The symmetry conditions for QEC are given by \Cref{thm:perfmeas}, where the order and daggers differ because of the $\mc{R}\ct$ in \Cref{eq:idealReducedMap}.
	We also denote the subset of a set of superoperators $\mf{A} \subset \bb{B(B(H))}$ that have a single Kraus operator by $\mf{A}_K$ so we can use the bijection $A \leftrightarrow \mc{A}$ in the following statement.
	
	\begin{cor}\label{thm:perfmeas}
		Let $U \in \bb{U}( \bb{H}, \bb{K})$, $R\in\mf{R}$, and $\mc{N} \in \bb{B(B(K))}$.
		The maps $\{\mc{\tilde{\dot{N}}}_U(A\ct RB) : \mc{A} \in \mf{L}_K(U, \mc{N}), \mc{B} \in \mf{L}_K(U)\}$ are logically degenerate.
		Furthermore, the maps $\{\mc{\tilde{\dot{N}}}_U(A\ct RB) : \mc{A} \in \mf{S}_K(U, \mc{N}, \mc{B} \in \mf{S}_K(U)\}$ are degenerate.
	\end{cor}
	
	A simple application of \Cref{thm:perfmeas} is to IID noise, that is, to superoperators $\mc{N}$ that can be written as $\mc{N} = \mc{N}_1^{\otimes n}$ for some noise process $\mc{N}_1 \in \bb{B(B(J))}$ acting on a single system with Hilbert space $\bb{J}$.
	For such noise, $\mf{L}(U, \mc{N}_1^{\otimes n})$ will typically contain permutation operators.
	Moreover, the recovery maps $R$ are typically chosen to be tensor products of elements of some ``nice error basis''\cite{Knill1996} $\mf{E} \subset \bb{U(J)}$ that contains the identity and can be used to construct the ideal PVM, which, for qubits, is taken to be the set of single-qubit Pauli matrices.
	We define the weight of some recovery map $R \in \mf{E}^{\otimes n}$ to be the number of subsystems on which it acts nontrivially.
	Recall that a group of permutations of $n$ objects is $k$-transitive if every ordered subset of $k$ objects can be mapped to every other ordered subset of $k$ objects.
	When a group is one-transitive, we say that it is transitive.
	We then have the following.
	
	\begin{cor}\label{thm:transitive}
		Let $n$ be a positive integer, $U \in \bb{U}(\bb{H}, \bb{K}^{\otimes n})$ be an encoding, $\mf{E}\subset\bb{U(K)}$ be a nice error basis, and $\mc{N} \in \bb{B(B(K}^{\otimes n}))$.
		Then, if ($\mf{L}(U, \mc{N})$) $\mf{S}(U, \mc{N})$ contains a $k$-transitive group, the set of weight-$k$ errors under ideal or noisy measurements will be partitioned into at most $\binom{|\mf{E}| + k - 2}{k}$ (logical) degeneracy classes.
	\end{cor}
	
	\begin{proof}
		From \Cref{thm:perfmeas}, for any permutation $\mc{P} \in (\mf{L}(U, \mc{N})) \: \:  \mf{S}(U, \mc{N})$ of the tensor factors of $\bb{K}^{\otimes n}$ and any $R \in \mf{E}^{\otimes n}$, the conditional maps $\mc{\dot{N}}_U(P^\dagger RP)$ and $\mc{\tilde{N}}_U(R)$ are (logically) degenerate.
		The number of distinct unordered combinations of length $k$ from $s$ items is given by $\binom{s + k - 1}{k}$. Then \Cref{thm:transitive} follows directly from \Cref{thm:perfmeas} and the fact that there are $\binom{|\mf{E}| + k - 2}{k}$ distinct unordered combinations of $k$ of the $|\mf{E}| - 1$ nontrivial errors.
	\end{proof}
	
	For example, let $\mf{E} = \bb{P} = \{I, X, Y, Z\}$ be the set of single-qubit Pauli operators.
	Then, if $\mf{L}(U, \mc{N})$ contains a one-transitive group, the weight one Pauli operators for IID noise are partitioned into at most three degeneracy classes.
	If $\mf{L}(U, \mc{N})$ contains a two-transitive group, then there will be at most $\binom{4}{2}=6$ degeneracy classes of weight-two Pauli errors.
	
	\section{Symmetries of stabilizer codes}\label{sec:stabSym}
	
	\begin{table*}
		\begin{tabular}{c | c | c | c }
			& Three-qubit code & Five-qubit code & Steane code \\
			\hline
			Generators & $ZZI$ & $XZZXI$ & $ZZZZIII$ \\
			& $IZZ$ & $IXZZX$ & $ZZIIZZI$ \\
			& & $XIXZZ$ & $ZIZIZIZ$ \\
			& & $ZXIXZ$ & $XXXXIII$ \\
			& &  & $XXIIXXI$ \\
			& & & $XIXIXIX$ \\
			\hline
			$\bar{X}$ & $X^{\otimes 3}$ & $X^{\otimes 5}$ & $X^{\otimes 7}$ \\
			$\bar{Z}$ & $Z^{\otimes 3}$ & $Z^{\otimes 5}$ & $Z^{\otimes 7}$ \\
		\end{tabular}
		\caption{Stabilizer generators and logical $X$ and $Z$ operators for common stabilizer codes.}
		\label{table:stabcodegens}
	\end{table*}
	
	As an application of \Cref{sec:equivalentConditionalMaps}, we now consider how symmetries can be used to accelerate simulations of stabilizer codes under IID noise.
	We focus on symmetries $\mc{A}$ and $\mc{B}$ that correspond to conjugation by unitary operators $A, B \in \bb{U(H)}$ with $A = B$ and typically set $\bar{A} = \bar{I}$.
	As we are considering IID noise, any permutation of the qubits in $\mf{L}(U)$ is an element of the symmetry group of $U$ under $\mc{N}$.
	We can then use logical operations and elements of the symmetry group to find sets of degenerate recovery maps.
	
	We begin by reviewing stabilizer codes.
	An $(n, k)$-stabilizer code is defined by a set $\{G_0, \ldots, G_{n-k-1}\}$ of $n-k$ distinct, independent, and commuting $n$-qubit Pauli operators, referred to as the stabilizer generators.
	By distinct, we mean that no pair of stabilizer generators differ by a phase.
	We can define a basis $\{\ket{\bar{z}}: z \in \mathbb{Z}_2^k\}$ of the $2^k$-dimensional subspace $\bb{H}$ stabilized by the stabilizer generators.
	The encoding is then
	\begin{align}
	U = \sum_{z \in \mathbb{Z}_2^k} \ketbra{\bar{z}}{z}.
	\end{align}
	The standard way to define such a basis is to choose $k$ mutually commuting independent Pauli operators $\bar{Z}_0, \ldots, \bar{Z}_{k-1}$ from outside of the stabilizer group that commute with each stabilizer generator to be the encoded $Z$ operators and set $\ket{\bar{z}} \in \bb{H}$ to be the simultaneous $+1$ eigenvector of $\bar{Z}_j^{z_j}$ for each $j = 0, \ldots, {k-1}$.
	That is, for each $z \in \mathbb{Z}_2^k$, $\ket{\bar{z}}$ is the unique $+1$ eigenvector of
	\begin{align}\label{eq:logicalBasis}
	\prod_{j \in \bb{Z}_k} \tfrac{1}{2}(I + \bar{Z}_j^{z_j}) \prod_{i = 0}^{n-k-1} \tfrac{1}{2}(I + G_i),
	\end{align}
	up to an overall phase.
	We can also define $k$ mutually commuting Pauli operators $\bar{X}_0, \ldots, \bar{X}_{k-1}$ that commute with each stabilizer generator and also satisfy $[\bar{X}_j, \bar{Z}_j] \neq 0$ and $[\bar{X}_j, \bar{Z}_l] = 0$ for all $l \neq j$ to be the encoded Pauli X operators.
	We can then define $\bar{Y}_j = i \bar{X}_j \bar{Z}_j$ and $\bar{\bb{P}} = \otimes_{j\in\bb{Z}_k} \{\bar{I}, \bar{X}_j, \bar{Y}_j, \bar{Z}_j\}$ in analogy with the physical Pauli operators.
	Any state in the code space can be written as
	\begin{align}\label{eq:logicalExpansion}
	\rho = \left(\sum_{\bar{P} \in \bar{\bb{P}}} \mu_P \bar{P}\right)  \prod_{i = 0}^{n-k-1} \tfrac{1}{2}(I + G_i),
	\end{align}
	where $\mu_P \in [-1,1]$.
	Stabilizer generators and $\bar{X}$ and $\bar{Z}$ operators for common $(n, 1)$-stabilizer codes are listed in \Cref{table:stabcodegens}, where we omit subscripts on the logical operators for $k=1$.
	
	We define the Pauli stabilizer group and the Pauli logical group of an encoding $U$ to be $\mf{S}_p(U) = \mf{S}(U) \cap \bb{P}^{\otimes n}$ and $\mf{L}_p(U) = \mf{L}(U) \cap \bb{P}^{\otimes n}$, respectively.
	Note that these groups are often simply referred to as \textit{the} stabilizer and logical groups respectively.
	For an $(n, k)$-stabilizer code, the Pauli stabilizer group is $\mf{S}_p(U) = \langle G_0, \ldots, G_{n-k-1}\rangle$ and the Pauli logical group is $\mf{L}_p(U) = \langle \mf{S}_p(U), \bar{Z}_0, \ldots, \bar{Z}_{k-1}, \bar{X}_0, \ldots, \bar{X}_{k-1} \rangle$.
	From \Cref{eq:logicalExpansion}, we see that $\mf{L}_p(U)$ forms a basis for the codespace of a stabilizer code.
	Then any permutation operator that permutes the elements of the Pauli stabilizer group and leaves the elements of $\bar{\bb{P}}$ invariant will be an element of the (general) stabilizer group $\mf{S}(U)$.
	Similarly, any permutation operator that permutes the elements of the Pauli stabilizer group and permutes the elements of the logical Pauli group $\mf{L}_p(U)$ will be an element of the (general) logical group $\mf{L}(U)$.
	
	Therefore, we can find permutation operators in the symmetry group of the corresponding code for IID noise and so partition the recovery operators into degeneracy classes using \Cref{thm:perfmeas} by considering only the action of permutations on the Pauli stabilizer and Pauli logical groups.
	In \Cref{table:stabcodegens} we list permutation operators that generate transitive groups for each code, and a two-transitive group for the Steane code.
	As $\bar{X}$ and $\bar{Z}$ are permutationally invariant for these codes, the permutation operators are in the stabilizer group $\mf{S}(U)$.
	
	\begin{table*}
		\begin{tabular}{c | c | c | c }
			& Three-qubit code & Five-qubit code & Steane code \\
			\hline
			Permutations & (0 1 2) & (0 1 2 3 4) & (3 4)(5 6) \\
			& & (0 4)(1 3) & (0 3 1)(2 4 5) \\
			\hline
			Transitivity & 1 & 1 & 2\\
			\hline
			$|\mf{R}|$ & 4 & 16 & 64 \\
			$|\mf{R}_{\rm IID}|$ & 2 & 4 & 5 \\
		\end{tabular}
		\caption{Permutations that leave the code space invariant for common stabilizer codes~\cite{Huang2019}.
			We also list the transitivity of the symmetry groups generated by these permutations, the number $|\mf{R}|$ of recovery maps, the number $|\mf{R}_{\rm IID}|$ of logically nondegenerate Pauli recovery maps for generic IID noise under the permutation group formed by the listed permutations.}
		\label{table:permutationgroups}
	\end{table*}
	
	\begin{table*}
		\begin{tabular}{c | c |c | c}
			Code & Symmetry operation & Logical operation & $|\mf{R}_{\rm dep}|$\\
			\hline
			Five-qubit & $\mc{Q}^{\otimes 5}$ & $\overline{Q}$ & 2\\
			\hline
			Steane & $\mc{Q}^{\otimes 7}$ &  $\overline{Q}$ & 3\\
			& $\mc{H}^{\otimes 7}$ &  $\overline{\mc{H}}$ &\\
		\end{tabular}
		\caption{An incomplete list of non-trivial operations which induce symmetries in some of the more popular quantum error correcting codes and the number $|\mf{R}_{\rm dep}|$ of logically nondegenerate Pauli recovery maps for local depolarizing noise under the listed symmetry operations and the permutation symmetries listed in \Cref{table:permutationgroups}. Note that these symmetries are valid under the conditions of \Cref{thm:symWithMeas}, and therefore require that the physical noise acting on the system commute with the symmetry operator. For example, we can have physical noise $\mc{N}=\mc{D}_p^{\otimes n}$, where $\mc{D}_p$ is a single qubit depolarizing channel with parameter $p$ because $[\mc{Q},\mc{D}_p]=[\mc{H},\mc{D}_p]=0$.}
		\label{table:logicalsymops}
	\end{table*}
	
	For $(n, k)$-stabilizer codes, it is common to consider only Pauli recovery maps.
	Pauli recovery maps for an $(n, k)$-stabilizer code with encoding $U$ can be written as $\{T L_T: T \in \bb{P}^{\otimes n} / \mf{L}_p(U)\}$, where any choice of Pauli operator $L_T \in \mf{L}_P(U)$ for each $T$ will define a valid set of recovery maps.
	The set $\bb{P}^{\otimes n} / \mf{L}_p(U)$ is sometimes referred to as the set of pure errors~\cite{Poulin2006}.
	For an $(n, k)$-stabilizer code, there are $2^{n - k}$ recovery maps, where typically $k \ll n$.
	Using the permutation operators, we can reduce the number of distinct conditional maps that need to be computed for IID noise using \Cref{thm:perfmeas}.
	Note that the choice of $L_T$ will not affect the number of logically degenerate recovery maps, however, it will change the number of degenerate recovery maps.
	This is because applying $L_T$ to a recovery map only alters the corresponding conditional map by a logical operation and keeps the same syndrome, so toggling $L_T$ will change the logical relation between elements of the same logical degeneracy class.
	
	For example, one could choose the recovery maps for the three-qubit code to be $\mf{R}_1=\{III, XII, IYI, IIY \}$, in which case there are three degeneracy classes under IID noise.
	We could also select $\mf{R}_2=\{III, XII, IXI, IIX \}$ which only has two degeneracy classes under IID noise.
	This apparent discrepancy occurs because $IYI$ and $IIY$ are logically degenerate to $IXI$ and $IIX$, respectively, which are degenerate to $XII$ under IID noise by \Cref{thm:perfmeas}.
	Both $\mf{R}_1$ and $\mf{R}_2$ have two logical degeneracy classes.
	A similar situation arises for the other codes, which explains the discrepancy between the seven nondegenerate recovery maps observed in Reference~\cite{Chamberland2017} and the five nondegenerate recovery maps proven for IID unitary noise in Reference~\cite{Huang2019}. Reference~\cite{Huang2019} speculated that the discrepancy was due to the restriction to noise models with a single Kraus operator.
	Instead, the discrepancy is due to the selection of recovery operations; in Reference \cite{Chamberland2017}, elements from five different equivalence classes were present in the set of recovery operations, whereas in Reference \cite{Huang2019} the set of recovery operations chosen contained elements of seven equivalence classes.
	
	If the noise commutes with additional logical operators, the number of logically nondegenerate recovery maps is further decreased.
	A common example is IID depolarizing noise $\mc{N} = \mc{D}_p^{\otimes n}$, where $\mc{D}_p(\rho) = p \rho + (1-p)I/2$ is the single-qubit depolarizing channel with parameter $p$.
	For any single-qubit unitary (or unital) channels $\mc{U}_0, \ldots, \mc{U}_{n-1}$, the composite channel $\otimes_j \mc{U}_j$ commutes with $\mc{D}_p^{\otimes n}$.
	In particular, let $Q = \sqrt{Z}\sqrt{X}$, which maps $X \to Y \to Z \to X$.
	For the five-qubit and Steane codes, we have $Q^{\otimes n} \in \mf{L}(U)$, where $\bar{Q}$ implements the same operation on the logical space, that is, it maps $\bar{X} \to \bar{Y} \to \bar{Z}$.
	The set $\mf{R}_{\rm IID}$ of representative elements for the degeneracy classes under IID noise can be chosen to be $\{I, X_0, Y_0, Z_0\}$ for the five-qubit code and $\{I, X_0, Y_0, Z_0, X_0 Z_1\}$ for the Steane code.
	We can map $Y_0$ and $Z_0$ to $X_0$ by applying powers of $Q^{\otimes n}$ and so the representative elements of the logical degeneracy classes under IID depolarizing noise can be chosen to be $\{I, X_0\}$ for the five-qubit code and $\{I, X_0, X_0 Z_1\}$ for the Steane code respectively.
	For the five-qubit code, $\dot{\mc{N}}_U(R) = \bar{\mc{D}}_p$ for any $R \in \{X_0, Y_0, Z_0\}$, so that the logical operations (that is, $\bar{Q}$ and $\bar{Q}^\dagger$) commute with $\dot{\mc{N}}_U(R)$ and so the representatives of the \textit{degeneracy} classes can also be chosen to be $\{I, X_0\}$.
	The same property does not hold for the Steane code.
	Some examples of symmetry operators for the three-, five-, and seven-qubit codes are given in \Cref{table:permutationgroups,table:logicalsymops}.
	Note that these operators can be used to generate equivalence classes whenever they commute with the physical noise map, that is, they do not necessarily require IID noise.
	We discuss symmetries in non-IID noise in \Cref{sec:noniid}.
	
	\subsection{Noisy readout measurements in stabilizer codes}\label{sec:noisymeas}\FloatBarrier
	
Syndrome measurements are used to project onto cospaces of stabilizer codes. A syndrome measurement is typically performed by coupling the encoded system to a readout system via a syndrome extraction circuit and then measuring the state of the readout system. The measurement outcome for a syndrome measurement is a bit string $m\in\bb{Z}_2^{n-k}$. A circuit diagram which implements such a measurement for a single syndrome bit $m_i\in\bb{Z}_2$ is given in \Cref{fig:synmeas}.
	
	\begin{center}
		\begin{figure}
			\includegraphics{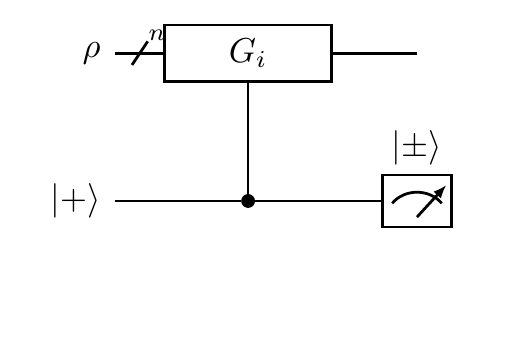}
			\vspace{-3em}
			\caption{Circuit diagram of a syndrome bit measurement corresponding to a generator $G_i$.}
			\label{fig:synmeas}
		\end{figure}
	\end{center}
	
	Let $\bb{K}$ and $\bb{T}$ be the Hilbert spaces of the encoded and readout systems respectively, $\mc{E} \in \bb{B}(\bb{K} \otimes \bb{T})$ be an ideal syndrome extraction circuit, and $\tilde{\mf{A}}$ be a noisy measurement acting on $\bb{K}\otimes\bb{T}$, where we have used $\: \tilde{}\: $ to denote noisy measurements or measurement operators.
	$\tilde{\mf{A}}$ can be implemented via a noisy X-basis measurement $\tilde{\mf{X}}:  \bb{T}\rightarrow \bb{R}$ by letting $\tilde{\mc{A}}=(\tilde{\mc{K}}\otimes \mc{I})\in\tilde{\mf{A}}, \forall \tilde{\mc{K}}\in\tilde{\mf{X}}$.
	Then the noisy measurement on the composite system is simply $\tilde{\mf{W}} = \{\tilde{\mc{A}}\mc{E} : \tilde{\mc{A}} \in \tilde{\mf{A}}\}$ and we can simplify the maps $\tilde{\mc{W}}\in\tilde{\mf{W}}$ to maps $\tilde{\Pi}_m\in\bb{B}(\bb{K})$ by using our knowledge of the initial state of the readout qubits.
	
	We can describe a noisy measurement by a confusion matrix (see \Cref{table:conf}).
	
	\begin{table}[h]
		\begin{tabular}{|c | c  |c|}
			\hline
			& $\ket{+}$ & $\ket{-}$\\
			\hline
			$\bra{+}$ & $a$& $1-a$\\
			\hline
			$\bra{-}$ & $1-b$ & $b$\\
			\hline
		\end{tabular}
		\caption{A confusion matrix for a noisy measurement in the X-basis where $a$ is the probability of correctly measuring $\ket{+}$ and $b$ is the probability of correctly measuring $\ket{-}$.}
		\label{table:conf}
	\end{table}
	Given a confusion matrix $C$, a noisy X-basis measurement $\tilde{\mf{X}}:\bb{T}\rightarrow\bb{R}$ can be expressed as a set of Kraus operators, $\{\tilde{X}_{i,j}\}$,
	\begin{align}
	\tilde{X}_{i,j}=\frac{1}{2}\sqrt{C_{i,j}}\bra{j},
	\end{align}
	where $i,j \in \{+,-\}$ with $i$ the observed outcome and $j$ the correct outcome. 
	We retrieve Kraus operators for the effective map on the encoded qubits by noting that the readout qubit starts in the state $\ket{+}\bra{+}$ so that we can write the Kraus operators $\{\tilde{P}_{m_i,j}\}$ of the noisy measurement $\tilde{\Pi}_{m_i,j}:\bb{K}\rightarrow\bb{K}$ with observed outcome $m_i$ and correct outcome $j$ as
	
	\begin{align}
	\tilde{P}_{G_i,m_i,j}=\frac{1}{2}\sqrt{C_{m_i,j}}\left(I\otimes \bra{j})cG_i(I\otimes\ket{+}\right).
	\end{align}
	Recalling that $cG_i=\ket{0}\bra{0}\otimes I+\ket{1}\bra{1}\otimes G_i$, we can simplify to
	
	\begin{align}
	\tilde{P}_{G_i,m_i,j}=\frac{1}{2}\sqrt{C_{m_i,j}}(I+(-1)^{\gamma(j)}G_i),
	\end{align}
	where $\gamma(+)= 0$ and $\gamma(-)= 1$.
	From here on, we use $\{0,1\}$ to label outcomes $\{+,-\}$ according to the mapping defined by $\gamma(\cdot)$.
	An ideal measurement is a special case of the above derivation, with $C_{i,j}=\delta_{i,j}$ so that the ideal measurement operators for a single syndrome bit are $P_{G_i,0}=(I+G_i)/2$ and $P_{G_i,1}=(I-G_i)/2$.
	Then we can express the Kraus operators for a noisy measurement in terms of the Kraus operators of an ideal measurement as follows:
	
	\begin{align}
	\tilde{P}_{m_i,j}=\frac{1}{2}\sqrt{C_{m_i,j}}P_{G_i,j}.
	\end{align}
	The noisy channel for the measurement of a single syndrome bit with observed outcome $m_i$ is the sum over a correct and incorrect measurement with outcome $m_i$,
	\begin{align}
	\tilde{\Pi}_{G_i,m_i}(\rho)&=\sum_{j\in\bb{Z}_2}\tilde{P}_{G_i,m_i,j}\rho\tilde{P}_{G_i,m_i,j}\\
	&=\sum_{j\in\bb{Z}_2}C_{m_i,j}P_{G_i,j}\rho P_{G_i,j}\\
	&=\sum_{j\in\bb{Z}_2}C_{m_i,j}\Pi_j(\rho).
	\end{align}
	A complete noisy syndrome measurement with outcome $m$ is implemented by $\tilde{\Pi}_m=\tilde{\Pi}_{m_{n-k}}\circ...\circ\tilde{\Pi}_{m_2}\circ\tilde{\Pi}_{m_{1}}$, so
	
	\begin{align}
	\tilde{\Pi}_m(\rho)&=\sum_{j\in\bb{Z}_2^{n-k}}\left(\prod_{i=1}^{n-k}\tilde{P}_{G_i,m_i,j_i}\right)\rho\left(\prod_{i=1}^{n-k}\tilde{P}_{G_i,m_i,j_i}\right)\\
	&=\sum_{j\in\bb{Z}_2^{n-k}}\left(\prod_{i=1}^{n-k}C_{m_i,j_i}^{(i)}P_{G_i,j_i}\right)\rho\left(\prod_{i=1}^{n-k}P_{G_i,j_i}\right),
	\label{eq:noisyprojunexpanded}
	\end{align}
	where $C^{(i)}$ is the confusion matrix for the $i^{th}$ readout qubit. Examining \Cref*{eq:noisyprojunexpanded}, we can see that the noisy projector onto a cospace associated with syndrome $m$ can be expressed as a probabilistic sum over ideal projections onto the cospaces of the code,
	
	\begin{align}
	\tilde{\Pi}_m(\rho)=\sum_{j\in\bb{Z}_2^{2^{n-k}}}\bar{C}_{j,m}\Pi_j(\rho),
	\end{align}
	where
	
	\begin{align}
	\bar{C}_{j,m}=\prod_{i=1}^{n-k} C_{m_i,j_i}^{(i)}
	\end{align}
	forms a confusion matrix for the logical measurement noise.
	
	\subsection{Symmetries of the toric code}\label{sec:toricsyms}\FloatBarrier
	
	\begin{figure}
		\includegraphics[width=\linewidth]{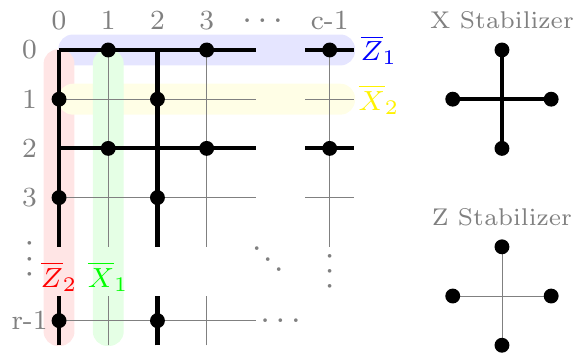}
		\caption{A toric code is specified by a number of rows, $r\in2\bb{Z}$, and a number of columns, $c\in2\bb{Z}$, with qubits lying on the edges of an $r\times c$ grid indexed by $\{(i,j): i\neq j, i\in\ \bb Z_r, j\in\bb Z_c\}$, as depicted above. The qubit in row $i$ and column $j$ is indexed by ($i$, $j$). The circles represent physical qubits, and the diagrams of X and Z stabilizers show which qubits these operators act on, with every qubit on the edge of an X (Z) stabilizer acted on by an X (Z) Pauli operator. The X stabilizers for the toric code are given by $\bb X=\{X(i,j): i\in 2\bb Z_{r/2}, j\in 2\bb Z_{c/2}\}$ and the Z stabilizers by $\bb Z=\{Z(i,j): i\in 2\bb Z_{r/2}+1, j\in 2\bb Z_{c/2}+1\}$, where $A(i,j)=A_{i-1,j}\otimes A_{i+1,j}\otimes A_{i,j-1}\otimes A_{i,j+1}$.
			The logical operators are $\bar{Z}_1=\bigotimes_{j\in 2\bb Z_{c/2}+1}Z_{0,j}$, $\bar{Z}_2=\bigotimes_{i\in 2\bb Z_{r/2}+1}Z_{i,0}$, $\bar{X}_1=\bigotimes_{j\in 2\bb Z_{r/2}}X_{i,1}$, and $\bar{X}_2=\bigotimes_{j\in 2\bb Z_{c/2}}X_{1,j}$.
			Note that all operations on the row (column) index are taken modulo $r$ ($c$), corresponding to periodic boundary conditions.
			The stabilizers are generated by $\bb X\cup\bb Z$, and a minimal generating group can be achieved by removing one X stabilizer and one Z stabilizer from $\bb X\cup\bb Z$.}
		\label{fig:toriccode}
	\end{figure}
	
	\begin{table*}
		\begin{tabular}{c | c | c | c}
			Code & Map name & Symmetry operation & Logical operation\\
			\hline
			$r\times c$ toric  & Twist & $(i,j)\rightarrow(i+2,j)\forall i,j$  & NA\\
			& Rotation & $(i,j)\rightarrow(i,j+2)\forall i,j$  & NA\\
			& Vertical reflection & $(i,j)\rightarrow(-i,j)\forall i,j$  & NA\\
			& Horizontal reflection & $(i,j)\rightarrow(i,-j)\forall i,j$  & NA\\
			& $\mc{J}$ & $\mc{H}^{\otimes n}\circ (i,j)\to(i+1,j+1)\forall i,j$ & $\overline{H}^{\otimes 2}\circ$ $\overline{\textrm{SWAP}}$\\
			$r\times r$ toric  & Diagonal reflection & $(i,j)\to(j,i)\forall i,j$ & $\overline{\textrm{SWAP}}$\\
		\end{tabular}
		\caption{Symmetry operations for the toric code, along with the associated logical operations. These operations can be applied for general IID noise, with the exception of $\mc{J}$, which requires noise that commutes with the $n$-qubit Hadamard gate as well as the permutation, for example, $\mc{J}$ can be applied to a toric code undergoing IID depolarizing noise.}
		\label{table:toric}
	\end{table*}
	
	We now show how our results can be applied to surface codes by considering the toric code as described in \Cref{fig:toriccode}.
	We do not fully specify an upper bound on the number of logically nondegenerate Pauli recovery maps because it depends on the number of rows and columns and involves high-weight recovery maps for large codes.
	Instead, we focus on weight-one and weight-two recovery maps.
	
	The torus is constructed from a rectangular lattice by identifying the top and bottom edges and then the left and right edges, which imposes periodic boundary conditions.
	Therefore, $\bar{Z}_1$ ($\bar{Z}_2$) can be moved to be any dark row (column) by multiplying the illustrated choice by $Z$ stabilizers along the rows (columns).
	Similarly, $\bar{X}_1$ ($\bar{X}_2$) can be moved to any light column (row) by multiplying the illustrated choice by $X$ stabilizers along the columns (rows).
	Therefore rotating or twisting the torus by two units to map dark lines to dark lines (horizontal or vertical translations with periodic boundary conditions) will permute the elements of $\mf{S}_p(U)$ and $\mf{L}_p(U)$, where the logical operation will simply be an identity, that is, twists that map $(i,j)\to (i+2,j)$ and rotations that map $(i,j)\to (i, j+2)$ are in the symmetry group of the $r\times c$ toric code.
	Using rotations and twists, we can map any weight-one Pauli to a weight-one Pauli that acts on one of the qubits at $(0, 1)$ or $(1, 0)$.
	Therefore, by \Cref{thm:perfmeas}, there will be at most 6 degeneracy classes of weight-one Pauli errors under IID noise instead of $3n$.
	
	We can also combine a rotation by 1, a twist by 1 (i.e. mapping all physical qubits by $(i,j)\to (i+1,j+1)$), and $H^{\otimes n}$, where $H$ is the Hadamard gate.
	This operation $\mc{J}$ preserves the stabilizer group by mapping $X$ stabilizers to $Z$ stabilizers and vice versa, $\bar{X}_1 \leftrightarrow \bar{Z}_2$, and $\bar{X}_2 \leftrightarrow \bar{Z}_1$, and so implements a logical SWAP combined with a Hadamard gate on each logical qubit.
	Therefore for noise that commutes with $\mc{J}$ (e.g., IID depolarizing noise), there are at most 4 logically nondegenerate weight-one Pauli recovery maps, namely, $\{X_{0,1}, Z_{0,1}, Y_{0, 1}, Y_{1, 0}\}$.
	
	From the periodic boundary conditions, we can also reflect vertically or horizontally by mapping $(i,j)\to (-i,j)$ or $(i,j)\to (i,-j)$, respectively.
	These reflections will permute the $X$ stabilizers and the $Z$ stabilizers and will either leave the logical operators invariant or map them to a different row or column, which is equivalent to the original logical operator up to a product of stabilizers.
	Therefore these reflections are elements of the (general) stabilizer group $\mf{S}(U)$.
	
	An $r\times r$ toric code can also be reflected across a diagonal axis via $(i,j)\to (j,i)$.
	This reflection permutes the stabilizers and maps $\bar{X}_1 \leftrightarrow \bar{X}_2$ and $\bar{Z}_1 \leftrightarrow \bar{Z}_2$ up to a product of stabilizers.
	Diagonal reflection is therefore an element of the logical symmetry group of an $r\times r$ toric code, which implements a logical SWAP gate.
	
	For any Pauli recovery map of a fixed weight-$w$, we can use rotations and twists to map one of the qubits $P$ acts nontrivially on to either $(0, 1)$ or $(1, 0)$.
	Therefore, the number of nondegenerate recovery maps is reduced by a factor of approximately $n / 2$, although the exact reduction factor introduced by translational symmetry depends upon $w$.
	We can further reduce the number of nondegenerate Pauli recovery maps using another symmetry of the toric code, namely, horizontal and vertical reflections about any row or column.
	Translational symmetries combined with horizontal and vertical reflection reduce the $9 \binom{n}{2}$ weight-two Pauli recovery maps to at most $9(\frac{n+c+r}{2}-2)$ nondegenerate recovery maps, which is approximately 9($n$-1)/2 up to edge effects.
	Using translation, we can map any weight-two Pauli recovery map to have weight on either $(0,1)$ or $(1,0)$.
	There are $n-1$ coordinate pairs containing each of these origins.
	Consider the coordinate pairs of the form $\{(0, 1), (i, j)\}$. Without loss of generality, we can use a vertical reflection---that is, $(i, j) \to (r-i, j)$---so that $i \in [0, r/2]$. We can then use a horizontal reflection and a rotation---that is, $(i, j)\to (i, 2 - j)$---so that $j \in [1, c / 2 + 1]$.
	There are then $\frac{\left(r/2+1\right)\times\left(c/2+1\right)}{2}-1$ coordinate pairs containing $(0,1)$ as we have reduced to a $(c/2+1)\times (r/2+1)$ grid with qubits on half of the locations, and we subtract the location of $(0,1)$.
	A similar reduction of coordinate pairs containing $(1,0)$ using horizontal reflection and vertical reflection with a twist of the form $(i,j)\to (2-i,j)$ produces pairs in $\{\{(1,0),(i,j)\}: i\in [1,r/2+1], j\in[0,c/2]\}$.
	Adding the sizes of these two sets of coordinate pairs, we get $\frac{rc/2+r+c}{2}-1$.
	There is one additional symmetry that occurs using these operations which is not covered by the above counting argument: There is a coordinate pair in each reduced set which, by undoing some of the operations used, maps to $\{(0,1),(1,0)\}$. As such, we can subtract one case.
	Then we multiply by 9 for the selection of an element in $\bb{P}_2$.
	Further reductions are possible using, e.g., a diagonal reflection for general IID noise in a square lattice or $\mc{J}$ for depolarizing noise.
	
	\subsection{Symmetries for non-IID noise}\label{sec:noniid}
	
	Though we have restricted attention in the examples thus far to IID noise in stabilizer codes, it should be noted that \Cref{thm:theMainThing}, \Cref{thm:symWithMeas}, and \cref{thm:perfmeas} apply equally to codes undergoing non-IID noise.
	In the case of non-IID noise, logically degenerate or strictly degenerate recovery operations can be calculated just as in previous sections.
	However, the form of the noise must be taken into account when finding (logically) degenerate maps, as commutation relations are not quite as simple for non-IID noise.
	If the non-IID noise commutes with a symmetry operator, it can be applied in the same way as in all of the examples thus far.
	For example, in a five-qubit code undergoing noise of the form $\mc{N}=\mc{N}_0\otimes\mc{N}_1\otimes\mc{N}_0\otimes\mc{N}_1\otimes\mc{N}_0$, the permutation operator $(0\: 4)(1\: 3)$ commutes with the noise and can thus be used to partition recovery operations, while $(0\: 1\: 2\: 3\: 4)$ does not preserve the noise and therefore cannot be used to generate symmetries.
	Then, using single-qubit $X$ operators as an example, $X_0\approxeq X_4 \neq X_1\approxeq X_3$, where we use $\approxeq$ to denote operators in the same degeneracy class, and $\neq$ to denote a separation between degeneracy classes.
	
	To find degeneracies when the physical noise does not commute with a symmetry operator, we must examine the effects of symmetries on the joint map $\mc{M}\mc{N}$, as in \Cref{thm:theMainThing}.
	For perfect measurements, this corresponds to finding degeneracies in maps of the form $\mc{R}\ct\mc{N}$, where $\mc{R}\ct$ is a recovery map.
	To find recovery maps which produce the same effective noise under a permutation operator, for example, we must simultaneously permute $\mc{N}$ and the recovery operation in question.
	This will not necessarily help in reducing computational complexity for a fixed noise model, but comes in handy when handling, for example, permutations of noise occurring for different error paths in a concatenated code.
	\Cref{sec:concatenated} gives a more complete example of how this can be applied, using a concatenated five-qubit code to illustrate, where the second level of concatenation sees non-IID noise even if the physical noise is IID.
	Note, however, that by the nature of concatenated codes, there are several permutations of noise at the second level of concatenation, allowing permutations to be applied more liberally to the noise model to find degeneracies between permutations of the noise as well as between recovery operations, whereas for a single noise model, any symmetry operator must preserve that noise model.
	
	\subsection{Symmetries of concatenated stabilizer codes}\label{sec:concatenated}
	
	\begin{figure*}
		\includegraphics{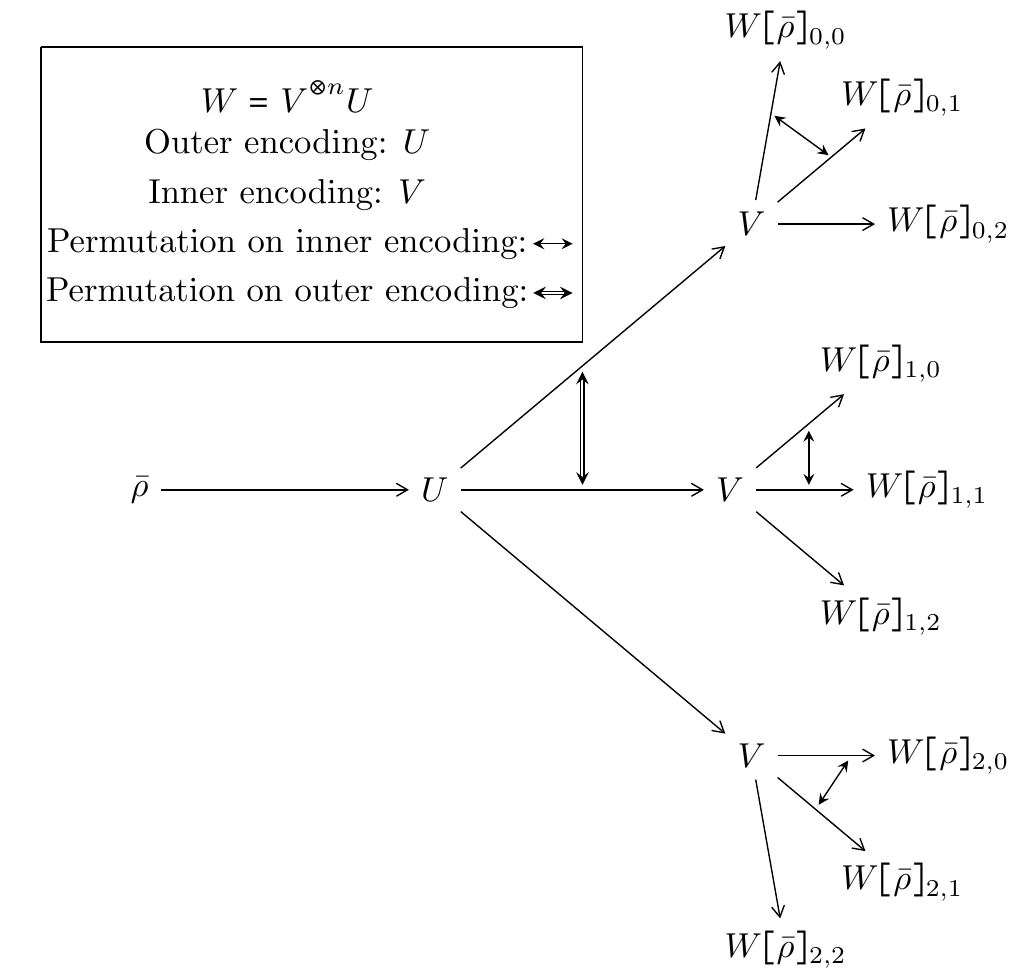}
		\caption{Tree structure of a concatenated code with outer encoding $U$ and inner encoding $V$. The leaves on the right represent physical qubits. The subscripts in $W[\bar{\rho}]_{i,j}$ denote that the physical qubit is the $j$th qubit in the inner encoding ($V$) of the $i$th qubit in the outer encoding ($U$). Double-ended arrows depict the action of a permutation which swaps the zeroth and first qubits on the inner and outer encodings.}
		\label{fig:concatenated_encoding}
	\end{figure*}
	
	A common method of improving the logical error rate is to concatenate QECCs.
	Let $U \in \bb{U}(\bb{H}, \bb{H}^{\otimes n})$ and $V \in \bb{U}(\bb{H}, \bb{H}^{\otimes m})$ be two encodings.
	Then $W = V^{\otimes n} U\in\bb{U}(\bb{H},\bb{H}^{\otimes mn})$ is a concatenated encoding with inner (outer) encoding $V$ ($U$), as illustrated in \Cref{fig:concatenated_encoding}.
	The number of recovery operators increases doubly exponentially in the number of levels of concatenation.
	However, we now show how \Cref{thm:perfmeas} and \Cref{thm:theMainThing} can be applied directly to concatenated codes to reduce the number of nondegenerate (or logically nondegenerate) recovery maps by concatenating symmetries of the inner and outer encodings.
	We apply \Cref{thm:perfmeas} at the first level as in previous examples, and in higher levels, use \Cref{thm:theMainThing} to apply symmetry operations.
	The reason \Cref{thm:theMainThing} is more applicable after the first encoding is that the effective physical noise for higher levels is no longer IID so the symmetry operations we've explored for IID noise no longer commute with the effective physical noises.
	We therefore must permute the noise and symmetry operations together to find degeneracy classes for different (noise map, recovery operation) pairs at each level.
	
	To analyze concatenated codes, we first need to construct a set of recovery maps and the logical group.
	Let $\mf{R}_U$ and $\mf{R}_V$ be sets of recovery maps for $U$ and $V$, respectively.
	Then $\mf{R}_W = \{(\otimes_j F_j) \mc{V}^{\otimes n}(G): F_0, \ldots, F_{n-1} \in \mf{R}_V, G \in \mf{R}_U \}$ is a set of recovery maps for $W$.
	To see this, let $R = (\otimes_j F_j) \mc{V}^{\otimes n}(G)$ and $\omega = (\otimes_j \xi_j) \mc{V}^{\otimes n}(\gamma)$ be two elements of $\mf{R}_W$ and recall that the recovery maps for, e.g., $U$ satisfy $U^\dagger Q^\dagger R U = \delta_{R, Q} 1_\bb{H}$.
	Then, as $\mf{R}_U$ and $\mf{R}_V$ are sets of recovery maps for $U$ and $V$, we have
	\begin{align}
	W^\dagger R^\dagger \omega W
	&= U^\dagger G^\dagger (\otimes_j V^\dagger F_j^\dagger \xi_j V) \gamma U \notag\\
	&= U^\dagger B^\dagger (\otimes_j \delta_{F_j, \xi_j} 1_\bb{H}) \gamma U \notag\\
	&= U^\dagger G^\dagger (1_\bb{H})^{\otimes n} \gamma U \prod_j \delta_{F_j, \xi_j} \notag\\
	&= 1_\bb{H} \delta_{G,\gamma},
	\end{align}
	as required.
	Any recovery maps for $U$ that are tensor products of logical operators for $V$ (e.g. Pauli recovery maps in concatenated stabilizer codes) can be commuted through $V^{\otimes n}$.
	
	We now show how to build some elements of the logical group $\mf{L}(W)$, focusing on tensor product operations and the symmetry group under IID noise.
	Suppose that $\mc{A}_0, \ldots, \mc{A}_{n-1} \in \mf{L}(V)$ and let $\mc{\bar{A}}_j$ be such that $\mc{A}_j \mc{V} = \mc{V}\mc{\bar{A}}_j$.
	Then, provided $(\otimes_j \mc{\bar{A}}_j) \in \mf{L}(U)$ with $(\otimes_j \mc{\bar{A}}_j) \mc{U} = \mc{U \hat{A}}$, we have
	\begin{align}
	(\otimes_j \mc{A}_j) \mc{W} &= (\otimes_j \mc{A}_j \mc{V) U}=(\otimes_j \mc{V\bar{A}}_j) \mc{U }\\
	&=\mc{ V}^{\otimes n} \mc{U \hat{A} = W\hat{A}},
	\end{align}
	and so $(\otimes_j \mc{A}_j) \in \mf{L}(W)$.
	
	The permutation group of a concatenated code can be generated from the permutation groups of the inner and outer encodings by labeling each qubit as a pair of indices and permuting one index with the inner code's permutation group and the other index by the outer code's permutation group.
	For example, let $\mc{A} \in \mf{L}(U)$ be a permutation of $n$ qubits.
	Then we can commute $\mc{A}$ through $\mc{V}^{\otimes n}$ by permuting the tensor factors as illustrated in \Cref{fig:concatenated_encoding}.
	
	If the permutation groups of both the inner and outer code in a concatenated code are transitive, then the concatenated code will also be transitive.
	This becomes evident when envisioning the action of permutation operators on the branches in \Cref{fig:concatenated_encoding}; if $U$ is transitive then each major branch (those to the right of $U$ but left of $V$) can be mapped to any other major branch.
	If $V$ is transitive, each subbranch of the major branches (those to the right of each $V$) can be mapped to any other subbranch of the same branch.
	Then it follows that any physical qubit can be mapped to any other physical qubits when the codes being concatenated are transitive, that is, the concatenation of transitive codes will also be transitive, although the concatenation of two-transitive codes will typically not be two-transitive because the individual permutations cannot separate errors that act on the same inner code block.
	
	For example, consider the nine-qubit Shor code, which can be regarded as a concatenated code with $U = \ketbra{000}{0} + \ketbra{111}{1}$ and $V = UH = \ketbra{000}{+} + \ketbra{111}{-}$.
	The Hadamard gate is introduced so that weight-one $Z$ errors, which are logical operators in the inner repetition code without the Hadamard, are mapped to $X$ errors and so can be detected and corrected by the outer code.
	There are $2^8 = 256$ recovery maps for the nine-qubit Shor code.
	Any permutation of the three qubits for either $U$ or $V$ is in both $\mf{S}(U)$ and $\mf{S}(V)$.
	Labelling the qubits by $(i, j)$ where $i$ labels the qubits in the inner encoding and $j$ labels the qubits in the outer encoding, we can apply any permutation to $j$ for each $i$ independently (i.e., permutations on the inner encoding) and any permutation on $j$ for all $i$ (i.e., permutations on the outer encoding, which affect all corresponding qubits in the inner encoding).
	Because both the inner and outer encodings are transitive, the weight-one Pauli recovery maps for the concatenated code fall into three degeneracy classes under IID noise with representative elements $X_{0, 0}$, $Y_{0, 0}$, and $Z_{0, 0}$.
	The weight-two Pauli recovery maps can be divided into two types.
	For the first type, both nontrivial Pauli terms act on a qubit in the same inner encoding.
	By permuting the outer encoding, we can map the nontrivial Pauli terms to act on the first inner encoding.
	We can then permute the qubits in the inner encoding so that the nontrivial Pauli terms act on qubits $(0, 0)$ and $(0, 1)$.
	Multiplying by $Z_{0, 0} Z_{0, 1}$ (a stabilizer) and permuting qubits $(0, 0)$ and $(0, 1)$ as necessary, there are 3 nondegenerate weight-two recovery maps of this type, namely, $P_{0, 0} X_{0, 1}$ for $P \in \{X, Y, Z\}$, where some of the other weight-two recovery maps are in the same degeneracy class as $I$ or weight-one errors.
	For the second type, one nontrivial Pauli term acts on a qubit in each of two distinct inner encodings.
	By permuting the inner encodings, we can map any such Pauli recovery map to $P_{0, 0} Q_{1, 0}$ for $P,Q\in\bb{P}$.
	There are nine such terms; however, we can swap the inner encodings to reduce to five terms (i.e., $ZX$ and $XZ$ are degenerate).
	Therefore there are at most five nondegenerate weight-two recovery maps of this type.
	Then there are at most 12 nondegenerate Pauli recovery maps under IID noise (1 of weight-zero, 3 of weight-one, and 8 of weight-two), compared to the total of 256 Pauli recovery maps selected for a given implementation.
	
	The reduction in the number of recovery maps is more dramatic for larger concatenated codes and higher levels of concatenation.
	For the five-qubit and Steane codes, there are $16^5 \approx 10^6$ and $64^7 \approx 4\times 10^{12}$ recovery maps at the first level of concatenation respectively.
	To quickly remove many degenerate recovery maps, it is more convenient to use the recursive structure typically used in numerical studies.
	For noise with $\mc{N} = \mc{M}^{\otimes m}$ and a recovery map $R = (\otimes_j F_j) \mc{V}^{\otimes n}(G)$, where $F_0, \ldots, F_{n-1} \in \mf{R}_V$ and $G \in \mf{R}_U$, we then have
	
	\begin{align}\label{eq:concatenatedDerivation}
	\dot{\mc{N}}_W(R)
	&= \mc{W}\ct \mc{R}\ct \mc{M}^{\otimes m} \mc{W} \notag\\
	&= \mc{U}\ct \mc{G} \left(\otimes_j \mc{V}\ct \mc{F}_j\ct \mc{M} \mc{V}\right) \mc{U} \notag\\
	&= \mc{U}\ct \mc{G} \left(\otimes_j \bar{\mc{M}}_V(\mc{F}_j)\right) \mc{U}.
	\end{align}
	That is, the logical noise map via \Cref{eq:reducedMap} for the concatenated code is simply the logical map for the outer code where the effective ``physical'' noise map is the effective logical noise for the inner code conditioned on the $F_j$.
	We then see that the number of logical nondegeneracy classes at the first level of a concatenated code will be at most $|\mf{R}_{V,\mc{M}}|^n \times |\mf{R}_U|$, where $\mf{R}_{V,\mc{M}}$ is the set of degeneracy classes under the noise process $\mc{M}$.
	For the five-qubit and Steane codes with IID depolarizing noise, there are then at most $2^5 \times 16 = 512$ and $3^7 \times 64 = 34,992$ logical degeneracy classes.
	
	We can further reduce the number of logical degeneracy classes by using the symmetry of the outer encoding and choosing the decoder correctly.
	Importantly, even if the physical noise is IID, the effective noise for the outer encoding conditioned on observed syndromes (i.e., the $F_j$ in \Cref{eq:concatenatedDerivation}) will not be IID.
	Nevertheless, the set of effective noise processes for the outer encoding is permutationally invariant with respect to the outer code.
	Consider two recovery maps $G \bigotimes_j F_j$ and $G' \bigotimes_j F'_j$.
	If $\bigotimes_j F'_j$ can be obtained from $\bigotimes_j F_j$ by an element of the symmetry group, then an optimal decoder will set $G'$ to be the same permutation of $G$ as then the two recovery maps will have the same (and optimal) logical noise by \Cref{thm:perfmeas}.
	Therefore, instead of all the elements of $\mf{R}_\mc{M}^{\otimes n}$, we need only consider the representative elements under the symmetry group of $U$.
	
	For example, for the five-qubit code under IID depolarizing noise there are $2^5$ effective noise processes $\otimes_j \bar{\mc{M}}_V(F_j)$ corresponding to the choices of $(F_0, \ldots, F_4) \in \{I, X_0\}^5$ (where $I$ and $X_0$ act on the relevant code block).
	We denote the operator that acts as $X_0$ on the $i$th code block and identity on the other blocks by $\hat{X}_i$.
	With this notation, because of the (0 1 2 3 4) symmetry, we only need to consider eight effective noise processes for an optimal decoder, namely those conditioned on the following operators:
	\newtagform{eqt}{(\,}{),}
	\usetagform{eqt}
	\begin{align}
	\mf{F}_0 = \{I\} \tag{no errors}\\
	\mf{F}_1 = \{\hat{X}_2\} \tag{one error}\\
	\mf{F}_2 = \{\hat{X}_0 \hat{X}_4, \hat{X}_1 \hat{X}_3\} \tag{two errors}\\
	\mf{F}_3 = \{\hat{X}_1 \hat{X}_2 \hat{X}_3, \hat{X}_0 \hat{X}_2 \hat{X}_4\} \tag{three errors} \\
	\mf{F}_4 = \{\hat{X}_0 \hat{X}_1 \hat{X}_3 \hat{X}_4\} \tag{four errors} \\
	\mf{F}_5 = \{\hat{X}_0 \hat{X}_1 \hat{X}_2 \hat{X}_3 \hat{X}_4\} \tag{five errors}
	\end{align}
	\usetagform{default}
	where we have chosen as representative elements the combinations that are invariant under reflection about qubit 2 (that is, the permutation $(0\: 4)(1\: 3)$, which is in the symmetry group of the five-qubit code).
	Furthermore, for IID depolarizing noise at the physical level (i.e., $\mc{M} = \mc{D}_p^{\otimes 5}$), each of the $\bar{\mc{M}}(\mc{F}_j)$ is a depolarizing channel and so commutes with $\mc{Q}$, so we need only consider the six choices of $G \in \{I, X_0, \ldots, X_4\}$.
	Moreover, by the reflection symmetry of the five-qubit code and the chosen combinations about qubit 2, we need only consider $G \in \{I, X_0, X_1, X_2\}$.
	That is, interpreting $(F, G)$ as specifying the recovery map $R = F \mc{V}^{\otimes n}(G)$, there are at most 32 degeneracy classes whose representative elements are
	\begin{align}
	(\mf{F}_0 \cup \mf{F}_1 \cup \mf{F}_2 \cup \mf{F}_3 \cup \mf{F}_4 \cup \mf{F}_5) \times \{I, X_0, X_1, X_2\}.
	\end{align}
	For the no-error or five-error cases, we can permute $G$ by arbitrary cyclic permutations because the effective noise for the outer encoding is IID and the outer encoding is invariant under cyclic permutations, so we need only consider $G \in \{I, X_1\}$ and so the elements in the two sets,
	\begin{align}
	\mf{F}_0 \times \{X_0, X_1, X_2\} \notag\\
	\mf{F}_5 \times \{X_0, X_1, X_2\}
	\end{align}
	are all degenerate, leaving only 28 degeneracy classes.
	We note that there is a further degeneracy, as explicit calculations show that there are, in fact, only 20 degeneracy classes as any two choices of $G$ that are related by a permutation that leaves $\otimes_j F_j$ invariant are degenerate.
	Writing pairs of noise maps $\mc{N}$ and recovery operations $R$ as ($\mc{N}$,$R$), we denote degeneracy by $(F, G) \cong (\xi, \gamma)$.
	From explicit calculations, we know that
	\begin{align}
	(\hat{X}_2, X_0) &\cong (\hat{X}_2, X_1), \notag\\
	(\hat{X}_0 \hat{X}_4, I) &\cong (\hat{X}_1 \hat{X}_3, I), \notag\\
	(\hat{X}_0 \hat{X}_4, X_1) &\cong (\hat{X}_0 \hat{X}_4, X_2), \notag\\
	(\hat{X}_1 \hat{X}_3, X_0) &\cong (\hat{X}_1 \hat{X}_3, X_2), \notag\\
	(\hat{X}_1 \hat{X}_2 \hat{X}_3, I) &\cong (\hat{X}_0 \hat{X}_2 \hat{X}_4, I), \notag\\
	(\hat{X}_0 \hat{X}_2 \hat{X}_4, X_0) &\cong (\hat{X}_0 \hat{X}_2 \hat{X}_4, X_2), \notag\\
	(\hat{X}_1 \hat{X}_2 \hat{X}_3, X_1) &\cong (\hat{X}_1 \hat{X}_2 \hat{X}_3, X_2),\textrm{ and} \notag\\
	(\hat{X}_0 \hat{X}_1 \hat{X}_3 \hat{X}_4, X_0) &\cong (\hat{X}_0 \hat{X}_1 \hat{X}_3 \hat{X}_4, X_1) .
	\end{align}
	However, we have not been able to identify an explicit symmetry to show that these cases are degenerate.
	
	\section{Conclusion}
	
	In this paper, we introduced notation which facilitates the exploration of the effective noise arising from a QECC, including the effects of noisy measurements.
	In \Cref{sec:noisymeas} we demonstrated that noisy readout measurements leave the effective logical noise invariant up to a renormalization for stabilizer codes.
	We showed how the computational cost of calculating the effective logical noise in a QECC can be reduced by orders of magnitude by identifying measurement outcomes and recovery maps which result in the same (or logically degenerate) noise in \Cref{sec:equivalentConditionalMaps}.
	This reduction in computational complexity does not reduce the accuracy of the simulation.
	In \Cref{sec:stabSym}, we demonstrated the usefulness of the reduction by presenting degeneracies for the three-, five-, and seven-qubit codes (\Cref{table:permutationgroups,table:logicalsymops}) as well as concatenated codes (\Cref{sec:concatenated}) and the toric code (\Cref{table:toric}) for IID noise.
	Identifying additional symmetries in these and other QECCs is an open question, as is characterizing how symmetries arise when multiple steps of error correction occur before decoding.
	We anticipate that our results can be used to construct better soft decoders for concatenated codes, since logically degenerate recovery maps should simply be altered to make them degenerate.
	
	Furthermore, a significant barrier to the successful implementation of a quantum error correcting protocol in a large system is the time that the error correction step takes.
	The methods in this paper can be used to simplify the decoding step, since optimal recovery maps for a small number of nondegenerate syndromes can be pre-computed and cached.
	This allows other syndromes to be straightforwardly reduced to the nondegenerate syndromes and the recovery maps can be altered accordingly.
	While we focused on the toric code due to its translational symmetry, we expect that similar reductions will also hold for surface codes \cite{Fowler2012, Bravyi2017} and color codes \cite{Bombin2006,Bombin2015} even without translational symmetry, in part because these codes have more logical gates with simple (e.g., tensor product) structures.
	
	\section{Acknowledgements}
	
	The authors thank Andrew Doherty for helpful discussions.
	This research was undertaken thanks in part to funding from TQT, CIFAR, the Government of Ontario, and the Government of Canada through CFREF, NSERC, and Industry Canada.
	
	\bibliography{library}
	
\end{document}